\documentclass[12pt]{amsart}

\usepackage[utf8]{inputenc}
\usepackage[OT1]{fontenc}
\usepackage{amsfonts, amsmath, amssymb}
\usepackage{graphicx}
\usepackage{listings}
\usepackage{xcolor}
\usepackage{hyperref}
\usepackage{xcolor}
\usepackage{multirow}
\usepackage{tabularx,booktabs}
\usepackage{lipsum}
\usepackage{diagbox}
\usepackage{tablefootnote}

\newtheorem{definition}{Definition}
\newtheorem{theorem}{Theorem}

\newtheorem{corollary}{Corollary}
\newtheorem{lemma}{Lemma}

\providecommand{\pre}{\mathrm{Pre}}
\providecommand{\suf}{\mathrm{Suf}}
\newcolumntype{C}{>{\centering\arraybackslash}X} 

\providecommand{\bd}[1]{\boldsymbol{#1}}
\providecommand{\bb}[1]{\mathbb{#1}}

\newtheorem{construction}{Construction}

\def\qi#1 {\fbox {\footnote {\ }}\ \footnotetext { From Qi: {\color{red}#1}}}

\begin{document}
\title{On the maximum size of variable-length non-overlapping codes}

\author{Geyang~Wang}
\address{ Department of Electrical and Computer Engineering $\&$ Institute for Systems Research, University of Maryland, College Park, MD 20742, USA}
\curraddr{}
\email{wanggy@umd.edu}
\thanks{}

\author{Qi~Wang}
\address{Department of Computer Science and Engineering $\&$ National Center for Applied Mathematics Shenzhen, Southern University of Science and Technology, Nanshan District, Shenzhen, Guangdong 518055, China}
\curraddr{}
\email{wangqi@sustech.edu.cn}
\thanks{}

\subjclass[2010]{Primary 94A45; Secondary 94B25}
\keywords{non-overlapping code, variable-length, fixed-length} 

\date{}

\dedicatory{}

\maketitle              
\begin{abstract}
Non-overlapping codes are a set of codewords such that the prefix of each codeword is not a suffix of any codeword in the set, including itself. If the lengths of the codewords are variable, it is additionally required that every codeword is not contained in any other codeword as a subword. Let $C(n,q)$ be the maximum size of $q$-ary fixed-length non-overlapping codes of length $n$. The upper bound on $C(n,q)$ has been well studied. However, the nontrivial upper bound on the maximum size of variable-length non-overlapping codes of length at most $n$ remains open. 
In this paper, by establishing a link between variable-length non-overlapping codes and fixed-length ones, we are able to show that the size of a $q$-ary variable-length non-overlapping code is upper bounded by $C(n,q)$. Furthermore, we prove that the average length of the codewords in a $q$-ary variable-length non-overlapping codes is lower bounded by $\lceil \log_q \tilde{C} \rceil$, and is asymptotically no shorter than $n-2$ as $q$ approaches $\infty$, where $\tilde{C}$
denotes the cardinality of $q$-ary variable-length non-overlapping codes of length up to $n$.  
\end{abstract}
\section{Introduction}\label{sec-intro}
Motivated by applications for synchronization in communications, the study of {\em non-overlapping codes} (also called strongly regular codes, cross-bifix-free codes) dates back to the 1950s~\cite{Schu56}. A code $S \subseteq \cup_{n \ge2} \mathbb{Z}_q^n$ is called {\em non-overlapping} if the two conditions are satisfied: 1) the prefix of each codeword is not a suffix of any codeword in $S$, including itself; 2) for all distinct codewords $\mathbf{u,v} \in S$, $\mathbf{u}$ does not
contain $\mathbf{v}$ as a subword. If all the codewords in $S$ have the same length $n$, then the second condition above is automatically satisfied, and in this case, $S$ is called a {\em fixed-length} non-overlapping code; otherwise, it is called a {\em variable-length} non-overlapping code. 

The study of non-overlapping codes mainly focuses on deriving their bounds on the cardinality with respect to the parameters including the code length $n$ and the alphabet size $q$, and also constructing non-overlapping codes of large size close to the bounds. The first construction was given by Levenshtein~\cite{Leven65,Leven70,Gilbert60}, and was rediscovered in~\cite{WW00,BS04,Chee13}. Recently, non-overlapping codes have found important applications in DNA storage systems~\cite{Yazdi18,LY19}. For
more constructions on fixed-length non-overlapping codes, for example, see~\cite{Bajic14,Bern14,Blackburn15,Bern17,Barcu18}. In addition, we refer to~\cite{WW22,Blackburn22,Stan23,Qin23,Cai23} for a series of recent advances in non-overlapping codes and their extensions. 

Denote by $C(n,q)$ the maximum size of fixed-length non-overlapping codes over an alphabet of size $q$. The best-known upper bound given by Levenshtein~\cite{Leven70} is
\[
    C(n,q) \le   \left( \frac{n-1}{n}\right)^{n-1} \frac{q^n}{n}.
\]
Blackburn~\cite{Blackburn15} further showed the tightness of this bound if $n$ divides $q$.  In 2017, Bilotta~\cite{Bilotta17} defined variable-length non-overlapping codes and gave a binary construction by extending Levenshtein's construction. Later, the authors~\cite{WW22} proposed a generating functions approach in constructing $q$-ary non-overlapping codes for both the fixed-length and variable-length cases. 
Unlike the fixed-length case, for the maximum size of variable-length non-overlapping codes, it seems difficult to find a direct upper bound, and only recursive bounds were reported~\cite{Bilotta17,WW22}. More precisely, let $S \subseteq \cup_{i = m}^n \mathbb{Z}_q^i$ denote a variable-length non-overlapping code over $\mathbb{Z}_q = \{0,1,\ldots, q-1\}$, an alphabet of size $q$, and $S_i = S \cap \mathbb{Z}_q^i$ denote the set of codewords of length $i$ in $S$ for $m \leq i \leq n$. One direct way to bound $S$ is to sum up all the values of
$|S_i|$ together, and this leads to the trivial bound $|S| \leq \sum_{i=m}^n C(i,q)$~\cite{Bilotta17}. Intuitively, this suggests that a variable-length non-overlapping code with codeword length at most $n$ may possibly contain more codewords than a fixed-length non-overlapping code with code length $n$. To the best knowledge of the authors, the problem of deriving a nontrivial direct bound on the cardinality of variable-length non-overlapping codes remains open. 

In this paper, we first establish a new link between variable-length non-overlapping codes and fixed-length ones by showing that a variable-length non-overlapping code can always be extended to a fixed-length non-overlapping code in a systematic way. The cardinality of variable-length non-overlapping codes can thereby be upper bounded by that of fixed-length non-overlapping codes. Furthermore, we investigate the average length of codewords in a variable-length
non-overlapping codes, and prove that it is lower bounded by $\lceil \log_q \tilde{C} \rceil$, where $\tilde{C}$ denotes the cardinality of $q$-ary variable-length non-overlapping codes of length up to $n$.

The rest of this paper is organized as follows. In Section~\ref{sec-pre}, we introduce some notations and definitions. In Section~\ref{sec-bound}, we first give a systematic way to extend variable-length non-overlapping codes to fixed-length non-overlapping codes, and then bound the cardinality of variable-length non-overlapping codes. In Section~\ref{sec-avg}, we provide results on the minimum average length of codewords in variable-length non-overlapping codes. Finally, we conclude this paper in Section~\ref{sec-con}.

\section{Preliminaries}~\label{sec-pre}
Let $n,q$ both be integers larger than $1$. Throughout this paper, let $\mathbb{Z}_q = \{0,1,\dots,q-1\}$ be the $q$-ary alphabet. Then $q$-ary codewords are vectors over $\mathbb{Z}_q$, and sometimes for convenience, we write vectors as strings. A $q$-ary code is a set of codewords over $\mathbb{Z}_q$, and is called {\em variable-length} if its codewords have different lengths.  The size of a $q$-ary code $S$ is the number of codewords in $S$, and is denoted by $|S|$. For a codeword $\bd{s} \in S$, its length is denoted by $|\bd{s}|$. 

For each $\bd{s} = (s_1,s_2,\dots,s_n) \in \bb{Z}_q^n$, denote the prefix and suffix of $\bd{s}$ of length $k$ by $\pre(\bd{s},k)= (s_1,\dots,s_k)$ and $\suf(\bd{s},k) = (s_{n-k+1}, \dots, s_n)$, respectively, where $0 \le k \le n$. In particular, we define $\pre(\bd{s},0)$ and $\suf(\bd{s},0)$ to be the empty string. Define 
$$
\pre(\bd{s}) = \{\pre(\bd{s},k) : k = 1,\dots, n-1\},
$$
and 
$$
\suf(\bd{s}) = \{\suf(\bd{s},k) : k = 1,\dots, n-1\}
$$ 
as the set of all nontrivial prefixes and suffixes of the codeword $\bd{s}$, respectively. The concatenation of two codewords $\bd{u}$ and $\bd{v}$ is denoted by $\bd{u}||\bd{v}$.

\begin{definition}[Non-overlapping codes]\label{def-noc}
    Let $S$ be a subset of $\cup_{i=2}^n \mathbb{Z}_q^i$. Then $S$ is called {\em non-overlapping} if 
    \begin{enumerate}
      \item[1)] For all $\bd{u},\bd{v} \in S, \pre(\bd{u}) \cap \suf(\bd{v}) = \emptyset$ ($\bd{u}$ and $\bd{v}$ may be identical);
        \item[2)] For all distinct $\bd{u},\bd{v} \in S$ with $|\bd{u}| \le |\bd{v}|$, $\bd{v}$ does not contain $\bd{u}$ as a subword, i.e., $\bd{u} \ne (v_{j+1}, v_{j+2}, \dots, v_{j+|\bd{u}|})$ for $0 \leq j \leq |\bd{v}|-|\bd{u}|$.
    \end{enumerate}
\end{definition}

For example, $1001$ is overlapping since $\pre(1001,1) = \suf(1001,1) = 1$, and $\{1100,10\}$ is also overlapping since $10$ is a subword of $1100$. Clearly, the code $\{11000,11010\}$ is non-overlapping. A non-overlapping code $S \subseteq \cup_{i=2}^n \mathbb{Z}_q^i$ is called {\em maximal} (or {\em non-expandable}) if for any $x \in \cup_{i=2}^n \mathbb{Z}_q^i \setminus S$, $S \cup \{x\}$ is overlapping, and $S$ is called {\em maximum} if $|S| \ge |S'|$ for any other non-overlapping code $S' \subseteq  \cup_{i=2}^n \mathbb{Z}_q^i$. We define both maximal and the maximum non-overlapping codes for the fixed-length case in a similar way.

By definition, it is straightforward to see that non-overlapping codes are also prefix codes, as defined in the following.

\begin{definition}[Prefix codes]\label{def-pc}\cite{CT06}
  A code $S \subseteq \cup_{i=2}^{n} \mathbb{Z}_q^i$ is called a {\em prefix code} if $\bd{u} \notin \pre{(\bd{v})}$ for all two distinct codewords $\bd{u}, \bd{v} \in S$. 
\end{definition}

By Definition~\ref{def-pc}, in a prefix code, no codeword can be a prefix of any other codeword.

\section{Upper bound on the size of variable-length non-overlapping codes}\label{sec-bound}

In this section, we first establish a link between variable-length non-overlapping codes and fixed-length ones. By the link, a variable-length non-overlapping code can always be transformed to a fixed-length non-overlapping code. 
The upper bound of variable-length non-overlapping codes can thereby be derived by that of fixed-length codes.

The following is a general construction that extends the codewords of different lengths in variable-length non-overlapping codes to those of the same length. 

\begin{construction} \label{cons:variable-to-fix}
    Let $S \subseteq \cup_{i=2}^n \mathbb{Z}_q^i$ be a $q$-ary variable-length code with codeword length at most $n$.
    Define
    \[
        \tilde{S} = \cup_{\bd{s} \in S} \tilde{\bd{s}},
    \]
    where
\begin{equation}\label{eqn-s}
  \tilde{\bd{s}} = \{ \bd{s} || \suf(\bd{x}, n - |\bd{s}|) : \bd{x} \in S \textrm{ and } |\bd{x}| > n - |\bd{s}|\}.
\end{equation}
\end{construction}

In essence, the idea of Construction~\ref{cons:variable-to-fix} is to extend all codewords that have length less than $n$ to codewords of length $n$, by appending all possible suffixes of certain codewords in the code. In such a way, $\tilde{S}$ is a fix-length code of length $n$. 

The following result will be useful in the proof that $\tilde{S}$ is a fixed-length non-overlapping code if $S$ is non-overlapping.

\begin{lemma}\label{lem-uv}
  Suppose that $S \subseteq \cup_{i=2}^n \mathbb{Z}_q^i$ is a $q$-ary variable-length non-overlapping code. For two distinct codewords $\bd{u}, \bd{v} \in S$, we have $\tilde{\bd{u}} \cap \tilde{\bd{v}} = \emptyset$, where $\tilde{\bd{u}}, \tilde{\bd{v}}$ are defined as in Eq.~(\ref{eqn-s}). 
\end{lemma}

\begin{proof}
  Assume without loss of generality that $|\bd{u}| \ge |\bd{v}|$. By the definition in Eq.~(\ref{eqn-s}), if $\tilde{\bd{u}} \cap \tilde{\bd{v}} \ne \emptyset$, then the codeword $\bd{v}$ must be the same as $\bd{u}$ or must be identical with the first $|\bd{v}|$ symbols. The latter case means that $\bd{v}$ is a codeword of $\bd{u}$, and cannot happen since $S$ is non-overlapping. 
\end{proof}

\begin{theorem}\label{thm:variable-to-fix}
  The fixed-length code $\tilde{S}$ by {\rm Construction~\ref{cons:variable-to-fix}} is non-overlapping if $S$ is non-overlapping.
\end{theorem}

\begin{proof}
The distinct codewords $\bd{u}', \bd{v}' \in \tilde{S}$ are constructed by extending some $\bd{u},\bd{v} \in S$. 
By Lemma~\ref{lem-uv}, such $\bd{u}, \bd{v}$ are unique. 
If  $|\bd{u}| = |\bd{v}| = n$, then we directly have  $\bd{u}' = \bd{u}$, $\bd{v}'= \bd{v}$, and they are clearly non-overlapping.
Without loss of generality, we assume that $|\bd{u}| \ge |\bd{v}|$ from now on.

\begin{figure}[h!]
  \centering
  \includegraphics[width=0.9\textwidth]{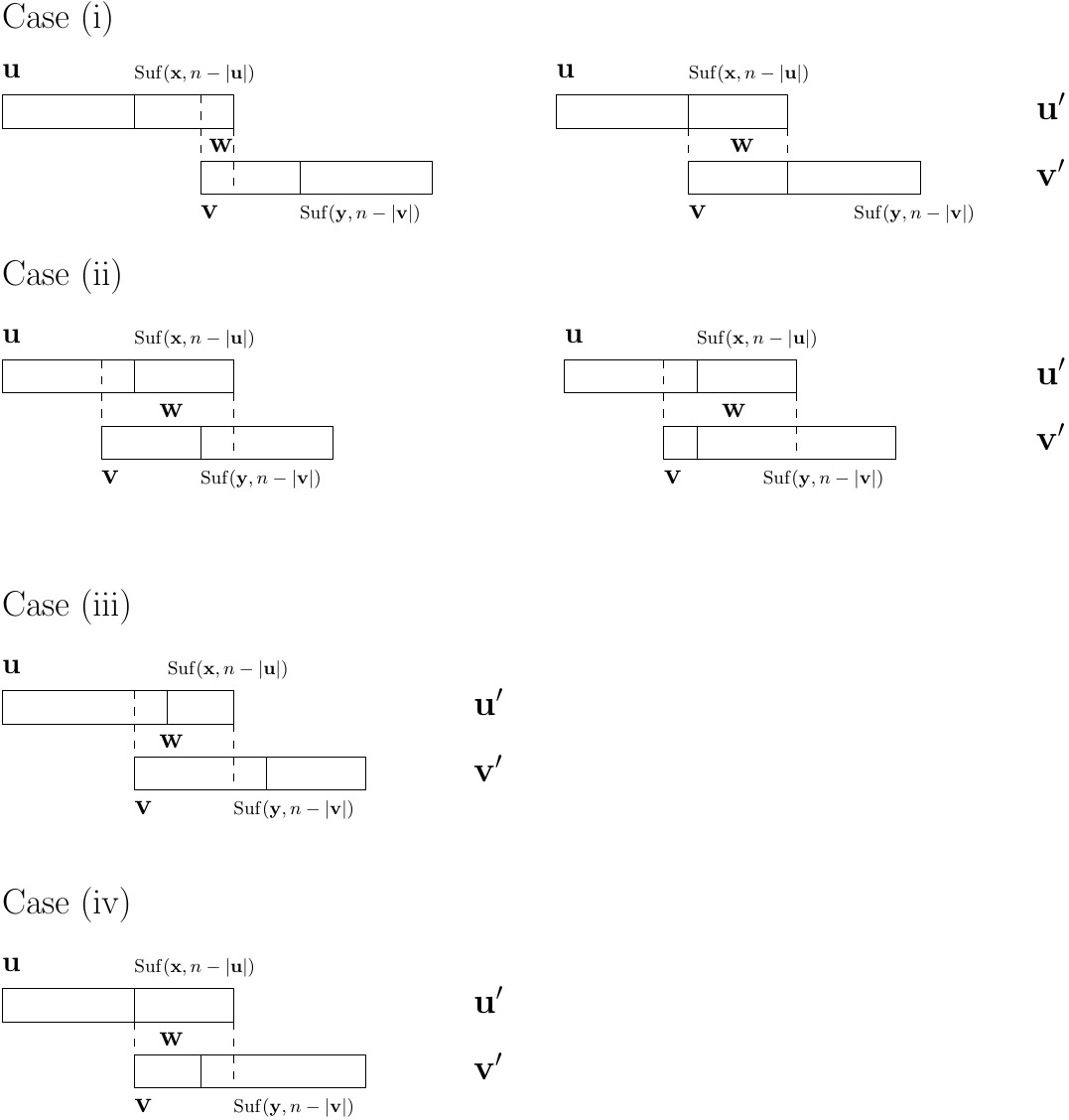}
  \caption{Four possible cases for $\bd{w} \in \pre(\bd{v}') \cap \suf(\bd{u}')$.}\label{fig:suf_u_pre_v}
\end{figure}

We first consider the case when $|\bd{u}|<n$.
It is clear that there uniquely exists a pair of codewords $\bd{x}, \bd{y} \in S$ such that $\bd{u'} = \bd{u} || \suf(\bd{x}, n - |\bd{u}|)$ and $\bd{v'} = \bd{v} || \suf(\bd{y}, n - |\bd{v}|)$. 
Since $\tilde{S}$ is now a code with all codewords having the same length $n$, by Definition~\ref{def-noc}, it remains to show that $\pre(\bd{v}') \cap \suf(\bd{u}') = \emptyset$ and $\pre(\bd{u}') \cap \suf(\bd{v}') = \emptyset$. 
We  check $\pre(\bd{v'}) \cap \suf(\bd{u'})$ first.  
Suppose that $\pre(\bd{v'}) \cap \suf(\bd{u'}) \ne \emptyset$, and we pick some word $\bd{w} \in \pre(\bd{v'}) \cap \suf(\bd{u'}) \ne \emptyset$.
As shown in Figure~\ref{fig:suf_u_pre_v}, there are four cases we need to discuss.

\begin{itemize}
  \item Case (i): $|\bd w| \le n - |\bd{u}|$, $|\bd w| \le |\bd v|$.
  
  If $|\bd w| < |\bd v|$, then $ \pre(\bd{v}, |\bd w|) = \suf(\bd x, |\bd w|)$. This is impossible, since this means that $\pre(\bd{v}) \cap \suf(\bd x) \ne \emptyset$, contradicting that $\bd v$ and $\bd x$ are non-overlapping. If $|\bd w| = |\bd v|$, then $\bd w = \bd v$ is a subword of $\bd x$, which also leads to a contradiction.

  \item Case (ii): $|\bd w| > n - |\bd u|$, $|\bd w| > |\bd{v}|$. 
  
  In this case, depending on the length of $\bd{v}$, either $\bd{v}$ is a subword of $\bd{u}$ or $\suf(\bd{u}, |\bd{w}| - (n - |\bd u|)) = \pre (\bd{v}, |\bd{w}| - (n - |\bd u|))$. 
  Both cases lead to a contradiction to the assumption that $\bd{u}, \bd{v}$ are non-overlapping.

  \item Case (iii): $|\bd w| > n - |\bd u|$, $|\bd w| \le |\bd v|$.
  
  In this case, $\suf(\bd{u}, |\bd{w}| - (n - |\bd{u}|))  = \pre(\bd{v}, |\bd{w}| - (n - |\bd{u}|))$, again a contradiction.
  
  \item Case (iv): $|\bd w| \le n - |\bd u|$, $|\bd w| > |\bd v|$.
  
  In this case, $\bd{v}$ must be a subword of $\bd x$, a contradiction.
\end{itemize}

\begin{figure}[h!]
  \centering 
  \includegraphics[width=0.9\textwidth]{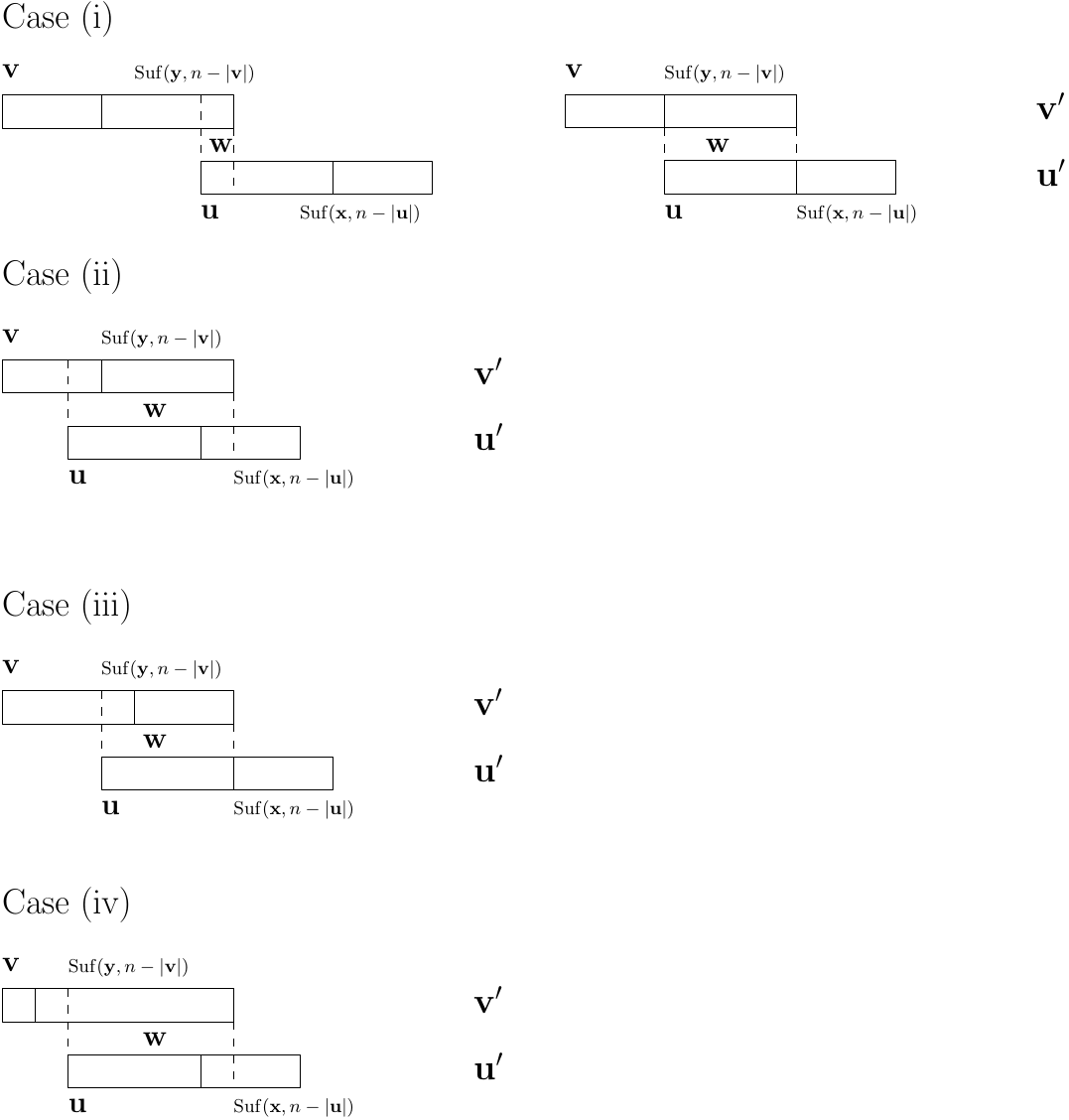}
  \caption{Four possible cases for $\bd{w} \in \pre(\bd{u}') \cap \suf(\bd{v}')$.}
  \label{fig:suf_v_pre_u}
\end{figure}

We now proceed to check $\pre(\bd u') \cap \suf (\bd v')$. 
Again, suppose that $\pre(\bd u') \cap \suf (\bd v') \ne \emptyset$, and we pick a certain word $\bd w \in \pre(\bd u') \cap \suf (\bd v')$. 
There are four cases to consider (see Figure~\ref{fig:suf_v_pre_u}).

\begin{itemize}
  \item Case (i): $|\bd w| \le n - |\bd v|$, $|\bd w| \le |\bd u|$.
  
  If $|\bd w| = |\bd u|$, then $\bd w = \bd u$ becomes a subword of $\bd y$.
  If $|\bd w| < |\bd u|$, then we have $\bd{w} = \suf(\bd{y}, |\bd w|) = \pre(\bd u, |\bd{w}|)$, which further implies that $\suf(\bd{y}) \cap \pre(\bd u) \ne \emptyset$.
  Therefore, both cases lead to a contradiction that the two codewords $\bd u$ and $\bd y$ in $S$ are non-overlapping.

  \item Case (ii): $|\bd w| > n - |\bd v|, |\bd w| > |\bd u|$.
  
  In this case, we have $\suf(\bd v, |\bd w| - (n - |\bd v|)) = \pre(\bd u, |\bd w| - (n - |\bd v|))$, leading to a contradiction that the two codewords $\bd u$ and $\bd v$ are non-overlapping.

  \item Case(iii):$|\bd w| > n - |\bd v|, |\bd w| \le |\bd u|$.
  
  As the same as Case (ii), we have $\suf(\bd v, |\bd w| - (n - |\bd v|)) = \pre(\bd u, |\bd w| - (n - |\bd v|))$, again a contradiction.

  \item Case (iv): $|\bd w| \le n - |\bd v|, |\bd w| > |\bd u|$.
  In this case, $\bd u$ becomes a subword of $\bd{y}$, a contradiction.
\end{itemize}
To sum up, we have both $\pre(\bd{v'}) \cap \suf(\bd{u'}) = \emptyset$ and $\pre(\bd u') \cap \suf (\bd v') = \emptyset$.

It remains to check the case when $|\bd{v}| < |\bd{u}| = n$. This can be done by a similar argument and is thus omitted.
Therefore, $\tilde{S}$ is non-overlapping.
\end{proof}

By the link established in Theorem~\ref{thm:variable-to-fix}, we are able to bound the size of variable-length non-overlapping codes in the following theorem. 

\begin{theorem}\label{thm:variableltfix}
    Let $S$ be a $q$-ary variable-length non-overlapping code with codeword length at most $n$. Let $C(n,q)$ denote the maximum size of a fixed-length $q$-ary non-overlapping code of length $n$. 
    Then we have 
    \begin{equation}\label{eqn-bound1}
        |S| \le C(n,q).
    \end{equation}
\end{theorem}

\begin{proof} 
  We denote by $S_i = S \cap \mathbb{Z}_q^i$ the set of codewords of length $i$ in $S$ for $m \leq i \leq n$, and by $\suf(S,i)$ the set of all possible suffixes of length $i$ of codewords in $S$.   
  By Construction~\ref{cons:variable-to-fix}, we have 
    \begin{eqnarray}\label{eqn-sizets}
       |\tilde{S}| & = & \sum_{\bd{s} \in S} |\tilde{\bd{s}}| = \sum_{i=m}^{n-1} \sum_{\bd{s} \in S_i} |\bar{\bd{s}}|  + |S_n| \nonumber \\
            & = & \sum_{i=m}^{n-1} \sum_{\bd{s} \in S_i} |\suf(S,n-i)| + |S_n| \\
	    & = & \sum_{i=m}^{n-1} |\suf(S, n-i)| |S_i| + |S_n|. \nonumber
    \end{eqnarray}
    
    Notice that $1 \le |\suf(S,k)| \le |\suf(S,k+1)|$ for $1 \leq k \leq n -1$, because every suffix of length $k$ is a subword of a certain suffix of length $k+1$.
    It then follows that 
    \[
        \sum_{i=m}^{n-1} |\suf(S, n-i)| |S_i| \ge \sum_{i=m}^{n-1} |\suf(S, 1)| |S_i|,
    \]
    and we further have 
    \begin{equation}\label{eqn-boundts}
      |\tilde{S}| \ge |S_n| + \sum_{i=m}^{n-1} |\suf(S, 1)| |S_i| \ge \sum_{i=m}^n |S_i| = |S|.
    \end{equation}
    By Theorem~\ref{thm:variable-to-fix}, $|\tilde{S}|$ is a fixed-length $q$-ary non-overlapping code of length $n$. Therefore, we have $|\tilde{S}| \le C(n,q)$, and this completes the proof of Eq.\eqref{eqn-bound1}. 
    
\end{proof}

\section{The minimum average length of variable-length non-overlapping codes}\label{sec-avg}

Now that by the bound of Eq.\eqref{eqn-bound1}, in general, variable-length non-overlapping codes cannot contain more codewords than fixed-length non-overlapping codes. It is then very natural to ask how long the average length of codewords in variable-length non-overlapping codes can be: 

``{\em Find the minimum average length $L$ of $q$-ary (variable-length) non-overlapping codes with cardinality $\tilde{C}$, where $C(n-1,q) < \tilde{C} \le C(n,q)$.}''

In this section, we address this problem and show that the minimum average length should be close to $n$ asymptotically.  Since non-overlapping codes are special prefix codes, we recall the following bound on the average length of a prefix code, which can also be applied to non-overlapping codes.


\begin{theorem}~\cite[Theorem~5.31]{CT06} \label{thm:L_low}
    Consider a prefix code for a random variable $X$ with codeword lengths $l_1, l_2, \ldots$ and corresponding probability $p_1, p_2, \ldots$. The expected length $L = \sum p_i l_i$ of any $q$-ary prefix code for a random variable $X$ satisfies 
    \[
        L \ge H_q(X),
    \]
    where $H_q(X) = -\sum_i p_i \log_q p_i $ is the $q$-ary entropy function of $X$, and equality holds if and only if $q^{-l_i} = p_i$.
\end{theorem}

By applying Theorem~\ref{thm:L_low} to non-overlapping codes, we have the following result. 

\begin{corollary}\label{cor:L_low}
  The average length $L$ of a $q$-ary non-overlapping codes with cardinality $\tilde{C}$ satisfies 
    \[
    L \ge \lceil \log_q \tilde{C} \rceil,
    \]
    where the equality holds if and only if each codeword is of length $\lceil \log_q \tilde{C} \rceil$.
\end{corollary}
\begin{proof}
  By viewing the non-overlapping code as a prefix code for a uniform random variable $X$ that takes values on $\{1,\dots,\tilde{C}\}$, the proof is completed.
\end{proof}

\begin{theorem}\label{thm:L}
  The minimal average length $L$ of a $q$-ary non-overlapping codes with cardinality $\tilde{C}$, $C(n-1,q) < \tilde{C} \le C(n,q)$ satisfies 
    \begin{equation}
      \lceil \log_q \tilde{C} \rceil \le L \le n,
    \end{equation}
    and 
    \begin{equation}
        n - 2 \le L \le n \quad \text{as } q \to \infty.
    \end{equation}
\end{theorem}

\begin{proof}
    We first have $L \le n$. 
    This is because by definition there exists a non-overlapping code with length $n$ and cardinality $C(n,q)$, and by dropping $C(n,q) - \tilde{C}$ codewords from that code forms a non-overlapping code with average length $n$.
    The lower bound $L \ge \lceil \log_q \tilde{C} \rceil$ comes from Corollary~\ref{cor:L_low}.

    The following construction for non-overlapping code is classic (see for example~\cite{Chee13}). 
    \begin{equation*}
        S = \{\bd{x} \vert \bd{x}_1 = 0, \bd{x}_i \ne 0, i = 2,\dots, n \}.
    \end{equation*}
    Therefore,
    \[
        C(n-1,q) \ge |S| = (q-1)^{n-2},
    \]
    and 
    \begin{equation*}
        \begin{split}
            \lceil \log_q \tilde{C} \rceil & \ge \lceil C(n-1, q) \rceil \\
                                    & \ge \lceil \log_q (q-1)^{n-2} \rceil  \\
                                    & = \lceil (n-2) \log_q (q-1) \rceil \\
                                    & = n - 2 \text{ as } q \to \infty.
        \end{split}
    \end{equation*}
\end{proof}

\section{Conclusion}\label{sec-con}
In this paper, we proved that the size of $q$-ary variable-length non-overlapping codes is upper bounded by $C(n,q)$, where $n$ is the length of the longest codeword, and $C(n,q)$ is the maximum size of fixed-length $q$-ary non-overlapping code of length $n$.
Furthermore, we investigate the minimal average length $L$ of variable-length non-overlapping codes and demonstrate that $n-2 \leq L \leq n$, when the cardinality of the code is between $C(n-1,q)$ and $C(n,q)$, and as $q$ tends to infinity.
These results suggest that variable-length non-overlapping codes do not offer advantages in terms of cardinality compared to fixed-length non-overlapping codes.




\begin{thebibliography}{10}

\bibitem{Bajic14}
Dragana Bajic and Tatjana Loncar-Turukalo, \emph{A simple suboptimal
  construction of cross-bifix-free codes}, Cryptography and Communications
  \textbf{6} (2014), 27--37.

\bibitem{BS04}
Dragana Bajic and Jakov Stojanovic, \emph{Distributed sequences and search
  process}, 2004 IEEE International Conference on Communications, vol.~1, IEEE,
  2004, pp.~514--518.

\bibitem{Barcu18}
Elena Barcucci, Antonio Bernini, Stefano Bilotta, and Renzo Pinzani, \emph{A
  2{D} non-overlapping code over a {$q$}-ary alphabet}, Cryptogr. Commun.
  \textbf{10} (2018), no.~4, 667--683. \MR{3770920}

\bibitem{Bern14}
Antonio Bernini, Stefano Bilotta, Renzo Pinzani, Ahmad Sabri, and Vincent
  Vajnovszki, \emph{Prefix partitioned gray codes for particular
  cross-bifix-free sets}, Cryptogr. Commun. \textbf{6} (2014), no.~4, 359--369.
  \MR{3258637}

\bibitem{Bern17}
Antonio Bernini, Stefano Bilotta, Renzo Pinzani, and Vincent Vajnovszki,
  \emph{A {G}ray code for cross-bifix-free sets}, Math. Structures Comput. Sci.
  \textbf{27} (2017), no.~2, 184--196.

\bibitem{Bilotta17}
Stefano Bilotta, \emph{Variable-length non-overlapping codes}, IEEE Trans.
  Inform. Theory \textbf{63} (2017), no.~10, 6530--6537.

\bibitem{Blackburn15}
Simon~R. Blackburn, \emph{Non-overlapping codes}, IEEE Trans. Inform. Theory
  \textbf{61} (2015), no.~9, 4890--4894. \MR{3386487}

\bibitem{Blackburn22}
Simon~R. Blackburn, Navid~Nasr Esfahani, Donald~L. Kreher, and Douglas~R.
  Stinson, \emph{Constructions and bounds for codes with restricted overlaps},
  arXiv preprint arXiv:2211.10309 (2022).

\bibitem{Cai23}
Qinlin Cai, Xiaomiao Wang, and Tao Feng, \emph{Constructions and bounds for
  $q$-ary $(1, k)$-overlap-free codes}, IEEE Transactions on Information Theory
  (2023).

\bibitem{Chee13}
Yeow~Meng Chee, Han~Mao Kiah, Punarbasu Purkayastha, and Chengmin Wang,
  \emph{Cross-bifix-free codes within a constant factor of optimality}, IEEE
  Trans. Inform. Theory \textbf{59} (2013), no.~7, 4668--4674.

\bibitem{CT06}
Thomas~M. Cover and Joy~A. Thomas, \emph{Elements of information theory},
  second ed., Wiley-Interscience [John Wiley \& Sons], Hoboken, NJ, 2006.
  \MR{2239987}

\bibitem{WW00}
A.J. De~Lind~van Wijngaarden and Tricia~J. Willink, \emph{Frame synchronization
  using distributed sequences}, IEEE Transactions on Communications \textbf{48}
  (2000), no.~12, 2127--2138.

\bibitem{Gilbert60}
E.~N. Gilbert, \emph{Synchronization of binary messages}, IRE Trans.
  \textbf{IT-6} (1960), 470--477. \MR{141545}

\bibitem{Leven65}
V.~I. Leven\v{s}te\u{\i}n, \emph{Decoding automata which are invariant with
  respect to the initial state}, Problemy Kibernet. \textbf{12} (1964),
  125--136.

\bibitem{Leven70}
V.~N. Leven\v{s}te\u{\i}n, \emph{The maximal number of words in codes without
  overlap}, Problemy Pereda\v{c}i Informacii \textbf{6} (1970), no.~4, 88--90.

\bibitem{LY19}
Maya Levy and Eitan Yaakobi, \emph{Mutually uncorrelated codes for {DNA}
  storage}, IEEE Trans. Inform. Theory \textbf{65} (2019), no.~6, 3671--3691.
  \MR{3959012}

\bibitem{Qin23}
Chunyan Qin, Bocong Chen, and Gaojun Luo, \emph{On non-expandable
  cross-bifix-free codes}, arXiv preprint arXiv:2309.08915 (2023).

\bibitem{Schu56}
M~Sch\"utzenberger, \emph{On an application of semi groups methods to some
  problems in coding}, IRE Transactions on Information Theory \textbf{2}
  (1956), no.~3, 47--60.

\bibitem{Stan23}
Lidija Stanovnik, Miha Mo{\v{s}}kon, and Miha Mraz, \emph{In search of maximum
  non-overlapping codes}, arXiv preprint arXiv:2307.12593 (2023).

\bibitem{WW22}
Geyang Wang and Qi~Wang, \emph{{$q$}-ary non-overlapping codes: a generating
  function approach}, IEEE Trans. Inform. Theory \textbf{68} (2022), no.~8,
  5154--5164.

\bibitem{Yazdi18}
S.~M. Hossein~Tabatabaei Yazdi, Han~Mao Kiah, Ryan Gabrys, and Olgica
  Milenkovic, \emph{Mutually uncorrelated primers for {DNA}-based data
  storage}, IEEE Trans. Inform. Theory \textbf{64} (2018), no.~9, 6283--6296.
  \MR{3849554}

\end{thebibliography}

\providecommand{\bysame}{\leavevmode\hbox to3em{\hrulefill}\thinspace}
\providecommand{\MR}{\relax\ifhmode\unskip\space\fi MR }
\providecommand{\MRhref}[2]{%
  \href{http://www.ams.org/mathscinet-getitem?mr=#1}{#2}
}
\providecommand{\href}[2]{#2}

\end{document}